\documentclass[pra, aps,letterpaper, twocolumn, preprintnumbers,superscriptaddress]{revtex4}
\usepackage{hyperref}
\usepackage{verbatim}
\usepackage{amsmath}
\usepackage{latexsym}
\usepackage{revsymb}
\usepackage{yfonts}

\usepackage{natbib}
\usepackage{amsfonts}
\usepackage{amsmath}
\usepackage{amssymb}
\usepackage{amsthm}
\usepackage{graphicx}
\usepackage{bm}
\usepackage{bbm}

\newcommand{\be}{\begin{eqnarray} \begin{aligned}}
\newcommand{\ee}{\end{aligned} \end{eqnarray} }
\newcommand{\benn}{\begin{eqnarray*} \begin{aligned}}
\newcommand{\eenn}{\end{aligned} \end{eqnarray*} }

\newcommand{\bc}{\begin{center}}
\newcommand{\ec}{\end{center}}

\newcommand{\id}{\mathbb{I}}

\newcommand{\tr}{\mathop{\mathsf{tr}}\nolimits}

\newtheorem{theorem}{Theorem}[section]

\newtheorem{lemma}[theorem]{Lemma}

\newcommand{\cS}{\mathcal{E}}
\newcommand{\cX}{\mathcal{X}}

\newcommand{\psuc}{p_{\rm succ}}
\newcommand{\psucPI}{p_{\rm succ}^{\rm PI}}
\newcommand{\vx}{{\vec{x}}}

\usepackage{amsfonts}

\def\id{\mathbb{I}}

\def\01{\{0,1\}}

\newcommand{\ket}[1]{|#1\rangle}

\newcommand{\proj}[1]{|#1\rangle\langle#1|}

\newcommand{\rank}{\operatorname{rank}}

\newcommand{\cB}{\mathcal{B}}
\newcommand{\cE}{\mathcal{E}}

\newenvironment{sdp}[2]{
\smallskip
\begin{center}
\begin{tabular}{ll}
#1 & #2\\
subject to
}
{
\end{tabular}
\end{center}
\smallskip
}

\newcommand{\ens}{\mathcal{E}}

\begin{document}

\title{Using post-measurement information in state discrimination}
\author{Deepthi \surname{Gopal}}
\email[]{deepthi@caltech.edu}
\affiliation{Institute for Quantum Information, Caltech, Pasadena, CA 91125, USA}
\author{Stephanie \surname{Wehner}}
\email[]{wehner@caltech.edu}
\affiliation{Institute for Quantum Information, Caltech, Pasadena, CA 91125, USA}

\date{\today}
\begin{abstract}
We consider a special form of state discrimination in which after the measurement we are given additional information
that may help us identify the state.
This task plays a central role in the analysis of quantum
cryptographic protocols in the noisy-storage model, where the identity of the state corresponds to a certain bit string, 
and the additional information is typically a choice
of encoding that is initially unknown to the cheating party.
We first provide simple optimality conditions for measurements for any such problem, 
and show upper and lower bounds on the success probability.
For a certain class of problems, we furthermore provide tight bounds on how useful post-measurement information can be.
In particular, we show that for this class finding the optimal measurement for the task of
state discrimination \emph{with} post-measurement information does in fact reduce to solving a 
different problem of state discrimination \emph{without} such information. 
However, we show that for the corresponding \emph{classical} state discrimination problems with post-measurement information such a reduction is impossible,
by relating the success probability to the violation of Bell inequalities. This suggests the usefulness of post-measurement information as another 
feature that distinguishes the classical from a quantum world.
\end{abstract}
\maketitle

\section{Introduction}

One of the characteristic traits of quantum mechanics is that not all possible states of a physical
system are perfectly distinguishable. This is in stark contrast to the classical world, but enables us
to solve cryptographic problems such as key distribution~\cite{bb84,e91} or two-party computation
in the noisy-storage model~\cite{noisy:new,prl:noisy}. Nevertheless, it is often possible to gain partial knowledge
about the state. Imagine a physical system is prepared in one out of several possible states chosen with a certain probability.
The set of possible states, as well as the distribution are thereby known to us. The goal of \emph{state discrimination} is to identify
which state was chosen by performing a measurement on the system, whereby our aim is to choose measurements that 
maximize the average probability of success. 
This fundamental problem has been studied extensively for the past 30 years, starting with the works of Helstrom~\cite{helstrom},  
Holevo~\cite{holevo} and Belavkin~\cite{belavkin:optimal} (see~\cite{sarah:survey} for a survey of known result), 
and has found many applications in quantum information theory (see e.g.,~\cite{ogawa:converse}), cryptography~\cite{GRTZ:qkd_review},
and algorithms~\cite{bacon:optimal, moore:pgm}.

\begin{figure}
\includegraphics[scale=0.35]{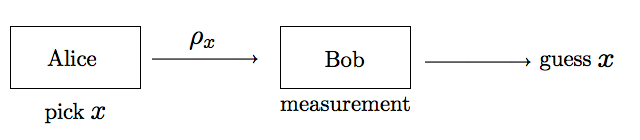}
\caption{Standard state discrimination}
\label{fig:stdgame}
\end{figure}

Here, we consider a special twist to the standard state discrimination problem introduced in~\cite{ww:pistar}, in which we obtain 
additional information after the measurement that may help us to identify the state. 
This task is easily described in terms of the following game depicted in
Figure~\ref{fig:game}: Imagine Alice chooses a state $\rho_{xb}$ from a finite set $\cS$ with probability $p_{xb}$,
labeled by what we will call the \emph{string} $x \in \cX$ and the \emph{encoding} $b \in \cB$. Bob knows $\cS$ as well as the distribution 
$P =  \{p_{xb}\}_{xb}$. Alice then sends the state to Bob. Bob may now perform any measurement
from which he obtains a classical measurement outcome $k$. Afterwards, Alice informs him about the encoding $b$. 
The task of \emph{state discrimination with post-measurement information} (and no memory) for Bob is to identify the string $x$, using the encoding $b$
and his classical measurement outcome $k$, where we are again interested in maximizing Bob's average probability of success over all
measurements he may 
perform~\footnote{Note that in~\cite{ww:pistar}, this problem was generalized to a setting where Bob may not 
only store classical information, but also a (limited) amount of quantum information. Here, however, we will only focus on the case of no storage
which was enough to relate security of the noisy-storage model to a coding problem~\cite{noisy:new}}. 
In~\cite{noisy:new} it was shown how bounds on this success probability can be used to prove security in the noisy-storage model. 

\begin{figure}
\includegraphics[scale=0.35]{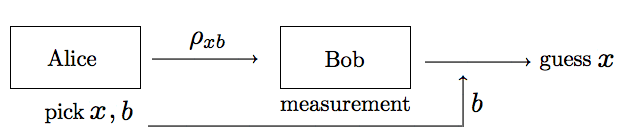}
\caption{Using post-measurement information}
\label{fig:game}
\end{figure}

Naturally, from a cryptographic standpoint it would be useful to know how much the additional information $b$ can actually help Bob. Let 
$\psucPI(\cS,P)$ and $\psuc(\cS,P)$ be the maximum average probabilities of success for the problem of state discrimination 
with and without post-measurement information respectively. 
Note that $\psucPI(\cS,P) \geq \psuc(\cS,P)$, since we can always choose to ignore any additional information.
We will measure how useful post-measurement information is for Bob in terms of the difference in his success
probability
\begin{align}
\Delta(\cS,P) := \psucPI(\cS,P) - \psuc(\cS,P)\ .
\end{align}
Of course, 
even in a classical setting post-measurement information can help Bob determine the string $x$. As a very simple example, suppose that
$x \in \{0,1\}$ is a single classical bit, and we have only one encoding $b \in \{0,1\}$. Imagine that Alice chooses $x$ and 
one of the two encodings uniformly at random and sends Bob the bit $x \oplus b = x + b \mod 2$. 
The states corresponding to this encoding are thus given by
\begin{align}
\rho_{xb} = \proj{x \oplus b}\ ,
\end{align}
where $p_{xb} = 1/4$.
Note that Bob now has a randomly chosen bit in his
possession and hence $\psuc(\cS,P) = 1/2$. However, he can decode correctly once he receives the additional information $b$ and
thus $\psucPI(\cS,P) = 1$, giving us $\Delta(\cS,P) = 1/2$.
As has been shown in~\cite{ww:pistar} we always have $\psucPI(\cS,P) = 1$ in the classical world where all states $\rho_{xb}$ are diagonal in the same
basis and orthogonal for fixed $b$. 

\subsection{Results}

We first provide a general condition for checking the optimality of measurements for our task (see Section~\ref{sec:optimality}). 
It was shown in~\cite{ww:pistar} that the optimal measurments can be found numerically using semidefinite programming solvers, however
in higher dimensions this remains prohibitively expensive.
We then focus on the case which is particularly interesting for cryptography, 
namely when the string $x$ is chosen uniformly and independently from the encoding $b$.
First, we provide upper and lower bounds for the success probability $\psucPI$ (Section~\ref{sec:lowerBound} and~\ref{sec:upperBound}).

In Section~\ref{sec:tightBounds}, we then show that for a large class of encodings (so-called \emph{Clifford encodings}) our lower bound is in fact tight. 
We thereby explicitely provide the optimal measurements for Clifford encodings.
The class of encodings we consider includes any encodings into two orthogonal pure states in dimension $d=2$
such as the well-known BB84 encodings~\cite{bb84}, as well as the case where we have two possible strings
and encodings which can be reduced to a problem in dimension $d=2$~\cite{halmos:blocks,ww:pistar}.  
It was previously observed that for BB84 encodings post-measurement information was useless~\cite{ww:pistar}. Here, we see that
this is no mere accident, and give a general condition for when post-measurement information is useless for Clifford encodings. 
We continue by showing that for Clifford encodings, we
can always perform a relabeling of the strings $x$ depending on the encoding $b$ such that we obtain a new problem for which post-measurement information is indeed useless.
This is particularly appealing from a cryptographic perspective as it means the adversary cannot gain any additional knowledge from the post-measurement information.
This means that for Clifford encodings we no longer need to treat the problem \emph{with} post-measurement information any differently, and can instead apply the well-studied
machinery of state discrimination. 

However, we will see that a relabeling that renders post-measurement information useless 
is impossible when considering a \emph{classical} ensemble~\footnote{An ensemble is classical if the states $\rho_{xb}$ all commute.}.
In particular, we will see that as long as we are able to gain 
some information about the encoded string $x$ without waiting for the post-measurement information,
then classically we cannot hope to find a non-trivial relabeling that makes post-measurement information useless. We thereby focus on the case of encodings a single bit into two possible
encodings in detail. Curiously, we will show this 
by relating the problem to Bell inequalities~\cite{bell}, such as for example the well-known
CHSH inequality~\cite{chsh}.
This suggests that the usefulness of post-measurement information forms another intriguing property that distinguishes the quantum from the classical world.

\section{General bounds}\label{sec:bounds}

Before investigating the use of post-measurement information, we derive general conditions for the optimality of measurements for our task.
We also provide a general bound on the success probability when the distribution over $\cX$ is uniform (i.e., $p_x = 1/|\cX|$) and independent of the choice
of encoding.

\subsection{SDP formalism}
When considering state discrimination with post-measurement information, we can without loss of generality assume that Bob performs
a measurement whose outcomes correspond to vectors
$\vec{x} = (x^{(1)},\ldots,x^{(L)}) \in \cX^{\times L}$ where each entry corresponds to the answer that Bob will give when he later learns
which one of the $L = |\cB|$ possible encodings was used. That is, when the encoding was $b$, Bob will output the guess $x^{(b)}$ of the vector $\vec{x}$~\cite{ww:pistar}.
In~\cite{ww:pistar} it was noted that the average probability that Bob outputs the correct
guess $x^{(b)}$ when given the post-measurement information $b$ maximized over all possible measurements (POVMs)
can be computed by solving the following semidefinite program (SDP). The primal of this SDP is given by
\begin{sdp}{maximize}{$v_{\rm primal} = \sum_{\vec{x}} \tr\left(M_{\vec{x}} \tau_{\vec{x}}\right)$}
& $\forall \vec{x} \in \cX^{\times L}, M_{\vec{x}} \geq 0$ \\
& $\sum_{\vec{x}} M_{\vec{x}} = \id$\ ,
\end{sdp}
where
\begin{align}\label{eq:avgState}
\tau_{\vec{x}} &= \sum_{b = 1}^{L} p_{x^{(b)} b}\ \rho_{x^{(b)}b}\ .
\end{align} 
By forming the Lagrangian, we can easily compute the dual of this SDP (see e.g.~\cite[Appendix A]{steph:diss}) which is given by
\begin{sdp}{minimize}{$v_{\rm dual} = \tr(Q)$}
& $\forall \vec{x} \in \cX^{\times L}, Q \geq \tau_{\vec{x}}$\ .
\end{sdp}
SDPs can be solved in polynomial time (in the input size) using standard algorithms~\cite{boyd:book}, 
which also provide us with the optimal measurement operators. 

\subsection{Optimality conditions}\label{sec:optimality}
However, with the SDP formalism in mind, it is now also easy to provide necessary and sufficient conditions for when a set of 
measurement operators $\{M_{\vec{x}}\}_{\vec{x}}$ 
is in fact optimal. Similar conditions were derived for the case of state discrimination \emph{without} post-measurement
information~\cite{croke:conditions,yuen:maxState,holevo:maxState,holevo:remarks,belavkin:optimal,belavkin:optimal2}. 
A proof can be found in the appendix.

\begin{lemma}\label{lem:conditions}
A POVM with operators $\{M_{\vec{x}}\}_{\vec{x}}$ is optimal for state discrimination with post-measurement information for the ensemble
$\cE = \{p_{xb}, \rho_{xb}\}$ if and only if the following two conditions hold:
\begin{enumerate}
\item $Q := \sum_{\vec{x}} \tau_{\vec{x}} M_{\vec{x}}$ is Hermitian.
\item $Q \geq \tau_{\vec{x}}$ for all $\vec{x} \in \cX^{\times L}$.
\end{enumerate}
\end{lemma}

\subsection{Upper bound}\label{sec:upperBound}

We now derive a simple upper bound on the success probability of state discrimination with post-measurement information when $p_{xb} = p_x p_b$ is a product
distribution, and the string $x$ is chosen uniformly at random (i.e.,$p_x = 1/|\cX|$). We will use a trick employed by Ogawa and Nagaoka~\cite{ogawa:converse} in the context of 
channel coding which was later rediscovered in the context of state discrimination~\cite{tyson:estimates}. A proof can be found in the appendix.

\begin{lemma}
Let $N = |\cX|$ be the number of possible strings, and suppose that the joint distribution over strings and encodings satisfies
$p_{xb} = p_b/N$, where the distribution $\{p_b\}_b$ is arbitrary. Then
\begin{align}
\psucPI(\cE,P) \leq \frac{1}{N} \tr\left[\left(\sum_{\vec{x}} \rho_{\vec{x}}^\alpha\right)^{1/\alpha}\right]\ ,
\end{align}
for all $\alpha > 1$, where $\cE = \{\rho_{xb}\}_{xb}$, $P = \{p_{xb}\}_{xb}$ and $\rho_{\vec{x}} = \sum_{b=1}^L p_b\ \rho_{x^{(b)} b}$.
\end{lemma}

Note that the bound on the r.h.s contains very many terms, and yet our normalization factor is only $1/N$. Nevertheless, for many interesting
examples we can obtain a useful bound this way, by choosing $\alpha$ to be sufficiently large.

\subsection{Lower bound}\label{sec:lowerBound}

Similarly, if $x$ is chosen uniformly at random and independent of the encoding, we can find a lower bound to $\psucPI$. 
The idea behind this lower bound is to subdivide the problem into a set of smaller problems which we can solve using standard
techniques from state discrimination.
Note that without loss of generality, we can label the elements of $\cX$ that we wish to encode from $0,\ldots,N-1$, where
we let $N = |\cX|$. The vector $\vec{x}$ can thus be written analogously as a vector $\vec{x} \in \{0,\ldots,N-1\}^{\times L}$.
We now partition the set of all possible such vectors as follows. Consider a shorter vector of length $L-1$, that is,
$\vec{y} \in \{0,\ldots,N-1\}^{\times (L-1)}$. With every such vector, we associate the partition
\begin{align}
T_{\vec{y}} &=
\{ \vec{x} = (y^{(1)} + j \mod N,\ldots,y^{(L-1)} + j \mod N,\\
&\qquad 0 + j \mod N)
 \mid j \in \{0,\ldots,N-1\}\}\ .\nonumber
\end{align}
Note that $|T_{\vec{y}}| = N$ and if $\vec{y} \neq \vec{\hat{y}}$ we have $T_{\vec{y}} \cap T_{\vec{\hat{y}}} = \emptyset$.
The union of all such partitions gives us the set of all possible vectors $\vec{x}$, that is,
\begin{align}
\bigcup_{\vec{y}} T_{\vec{y}} = \{\vec{x} \mid \vec{x} \in \{0,\ldots,N-1\}^{\times L}\}\ .
\end{align}
With every partition $T_{\vec{y}}$ we can now associate a standard state discrimination problem \emph{without} post-measurement information
in which we try to discriminate states 
\begin{align}\label{eq:avgXState}
	\rho_{\vec{x}} := \sum_{b=1}^L p_b \rho_{x^{(b)} b}\ ,
\end{align}
such that $\vec{x} \in T_{\vec{y}}$. 
That is, the set of states is given by $\ens_{T_{\vec{y}}} = \{\rho_{\vec{x}} \mid \vec{x} \in T_{\vec{y}}\}$ and $p_{\vec{x}} = 1/N$ is the uniform distribution.
Note that the original problem of state discrimination where we do not receive any post-measurement information corresponds to the partition given by
$\vec{y} = (0,\ldots,0)$, where we always give the same answer no matter what the post-measurement information is going to be. 
As we show in the appendix

\begin{lemma}\label{lem:lowerBound}
The success probability \emph{with} post-measurement information is at least as large as the success probability of 
a derived problem \emph{without} post-measurement information, i.e.,
$$
\psucPI(\ens,P) \geq \max_{\vec{y}} \psuc(\ens_{T_{\vec{y}}},P)\ .
$$
\end{lemma}

In particular, this allows us to apply any known lower bounds for the standard task of state discrimination~\cite{tyson:newIterate} to this problem.
Curiously, we will see that there exists a large class of problems for which this bound is tight, even though $\Delta(\ens,P) > 0$, that is, even though post-measurement
information is useful.

\section{Tight bounds for special encodings}\label{sec:tightBounds}

We now consider a very special class of problems called \emph{Clifford encodings},  for which we can determine the optimal measurement explicitly.
In this problem, we will only ever encode a single bit $x \in \{0,1\}$ chosen uniformly at random independent of the choice of encoding, 
and take $d = 2^{n}$ dimensional states of the form
\begin{align}
\rho_{xb} = \frac{1}{d}\left(\id + \sum_{j=1}^{2n+1} \gamma_{xb}^{(j)} \Gamma_j\right)\ , 
\end{align}
where $\Gamma_1,\ldots,\Gamma_{2n+1}$ are generators of the Clifford algebra, that is, anti-commuting operators~\footnote{That is $\{\Gamma_j,\Gamma_k\} = \Gamma_j \Gamma_k + \Gamma_k \Gamma_j = 0$ for $j \neq k$.}
satisfying $(\Gamma_j)^2 = \id$ for all $j$. We also assume that the vector 
$\gamma_{xb} = (\gamma_{xb}^{(1)},\ldots,\gamma_{xb}^{(2n+1)})$ satisfies $\gamma_{xb} = - \gamma_{(1-x)b}$ and
$\|\gamma_{xb}\|_2 \leq 1$. The distribution over encodings can be arbitrary.
Using the fact that the operators anti-commute, it is not hard to see that $\tr(\Gamma_j \Gamma_k) = 0$ for $j \neq k$ and the latter condition then ensures that
$\rho_{xb}$ is a valid quantum state~\cite{ww:cliffordUR}, that is,
$\rho_{xb}$ is positive semi-definite satisfying $\tr(\rho_{xb}) = 1$.
The Clifford algebra has a unique representation by
Hermitian matrices on $n$ qubits (up to unitary equivalence) which we fix henceforth.
This representation can be obtained via the famous Jordan-Wigner transformation~\cite{JordanWigner}:
\begin{align*}
  \Gamma_{2j-1} &= Y^{\otimes(j-1)} \otimes Z \otimes \id^{\otimes(n-j)}, \\
  \Gamma_{2j}   &= Y^{\otimes(j-1)} \otimes X \otimes \id^{\otimes(n-j)},
\end{align*}
for $j=1,\ldots,n$, where we use $X$, $Y$ and $Z$ to denote the Pauli matrices.
We also use $\Gamma_{2n+1} = i \Gamma_1\ldots\Gamma_{2n}$.

Note that in dimension $d=2$, these operators are simply the Pauli matrices
$\Gamma_1 = Z$, $\Gamma_2 = X$ and $\Gamma_{2n+1)} = Y$ and \emph{any} encoding of the bit $x$ into two orthogonal pure states
is of the above form.
A simple example, is the BB84 encoding~\cite{bb84} where we encode the bit $x$ into the computational basis labeled by $b = 0$ and into the Hadamard basis labeled by $b=1$.
Furthermore, if we have only two possible strings and encodings, we can always reduce the problem to dimension $d=2$~\cite{halmos:blocks,ww:pistar}.
In higher dimensions, encodings of the above form were suggested for the use in cryptographic protocols~\cite{ww:cliffordUR}.

\subsection{Without post-measurement information}

We now first examine the setting of state discrimination \emph{without} post-measurement information, which will provide
us with the necessary intuition. Again, we use $L = |\cB|$ to denote the number of possible encodings. 
Recall the average state $\rho_{\vec{x}}$ from~\eqref{eq:avgXState}
for the vector $\vec{x} = (x^{(1)}, \ldots,x^{(L)})$, which tells us for every possible encoding which bit appears in the sum.
We furthermore define the complementary vector $\vec{\underline{x}} = ((1-x^{(1)}),\ldots,(1 - x^{(L)}))$, that is, $\vec{x} + \vec{\underline{x}} = 0$.
As a warmup, suppose we are given $\rho_{\vec{x}}$ and $\rho_{\vec{\underline{x}}}$ chosen uniformly at random and wish to determine which one. 
Clearly, this is an example of state discrimination \emph{without} post-measurement information, which can also be written as an SDP~\cite{yuen:maxState,eldar:sdp}. The primal is
of the form
\begin{sdp}{maximize}{$\frac{1}{2}\left(\tr(M_{\vec{x}}\rho_{\vec{x}}) + \tr(M_{\vec{\underline{x}}} \rho_{\vec{\underline{x}}})\right)$}
&$M_{\vec{x}} \geq 0$\ ,\\
&$M_{\vec{\underline{x}}} \geq 0$\, \\
&$M_{\vec{x}} + M_{\vec{\underline{x}}} = \id$\ .
\end{sdp}
Its dual is easily found to be
\begin{sdp}{minimize}{$\tr(Q)$}
	&$Q \geq \frac{1}{2}\rho_{\vec{x}}$\ ,\\
	&$Q \geq \frac{1}{2}\rho_{\vec{\underline{x}}}$\ .
\end{sdp}
Analogous to Lemma~\ref{lem:conditions} with $\tau_{\vec{x}} = \frac{1}{2} \rho_{\vec{x}}$ one can derive optimality conditions which for the case of state discrimination
were previously obtained in~\cite{croke:conditions,yuen:maxState,holevo:maxState,holevo:remarks,belavkin:optimal,belavkin:optimal2}. In our case they tell us that
$Q = \frac{1}{2}(\rho_{\vec{x}} M_{\vec{x}} + \rho_{\vec{\underline{x}}} M_{\vec{\underline{x}}})$ must be Hermitian, and $Q$ is a feasible dual solution.
All we have to do is thus to guess an optimal measurement, and use these conditions to prove its optimality. 
Consider the operators 
\begin{align}
\label{eq:meas}
M_{\vec{x}} &= \frac{1}{2}\left(\id + \sum_j a^{(j)}_{\vec{x}} \Gamma_j\right)\ ,\\
M_{\vec{\underline{x}}} &= \frac{1}{2}\left(\id - \sum_j a^{(j)}_{\vec{x}} \Gamma_j\right)\ ,\nonumber
\end{align}
where $\vec{a}_{\vec{x}} = \vec{v}_{\vec{x}}/\|\vec{v}_{\vec{x}}\|_2$
is the normalized average vector
\begin{align}
\vec{v}_{\vec{x}} = \sum_{b=1}^L p_b \gamma_{x^{(b)}b}\ .
\end{align}
Note that since the generators of the Clifford algebra anti-commute, we have that
$M_{\vec{x}}, M_{\vec{\underline{x}}} \geq 0$ and $M_{\vec{x}} + M_{\vec{\underline{x}}} = \id$.
Hence, these operators do form a valid measurement. In the appendix, we derive two lemmas 
which show that 
$Q  = \frac{1}{2}(\rho_{\vec{x}} M_{\vec{x}} + \rho_{\vec{\underline{x}}} M_{\vec{\underline{x}}})$ is Hermitian (Lemma~\ref{lem:Qsum}) and
satisfies $Q \geq \frac{1}{2}\rho_{\vec{x}}$ for all $\vec{x}$ (Lemma~\ref{lem:Qsum} and~\ref{lem:eigenvalues})~\footnote{Recall that
for any Hermitian operator we have $\lambda_{\rm max}(A) \id \geq A$, where $\lambda_{\rm max}(A)$ is the largest eigenvalue of $A$.}, 
which are the conditions we needed for optimality.
All proofs can be found in the appendix.

\begin{theorem}
The measurements given in~\eqref{eq:meas} are optimal to discriminate $\rho_{\vec{x}}$ from $\rho_{\vec{\underline{x}}}$ chosen
with equal probability.
\end{theorem}

\subsection{With post-measurement information}

We are now ready to determine the optimal measurements for the case \emph{with} post-measurement information.
First of all, recall from Lemma~\ref{lem:lowerBound} that we can subdivide our problem into smaller parts by partitioning the set of
strings $\vec{x}$. Applied to the present case, these partitions
are simply given by
\begin{align}
\tilde{T}_{\vec{x}} = \{\vec{x}, \vec{\underline{x}}\}\ ,
\end{align}
where for simplicity we here use the vector $\vec{x}$ itself to label the partition.
Note that by Lemma~\ref{lem:lowerBound} we thus have that
\begin{align}\label{eq:rhs}
\psucPI(\ens,P) \geq \max_{\vec{x}} \psuc(\ens_{\tilde{T}_{\vec{x}}})\ .
\end{align}
We show in the appendix that this bound is in fact tight. 
\begin{lemma}\label{lem:tightBound}
For Clifford encodings
\begin{align}
\psucPI(\ens,P) = \max_{\vec{x}} \psuc(\ens_{\tilde{T}_{\vec{x}}})\ ,
\end{align}
and post-measurement information is useless if and only if the maximum on the r.h.s. is attained by $\vec{x} = (0,\ldots,0)$.
\end{lemma}

Note that the optimal measurement is thus given by~\eqref{eq:meas} for the vector $\vec{x}$ maximizing the r.h.s of~\eqref{eq:rhs},
and letting all other $M_{\vec{\tilde{x}}} = 0$. This shows that for our class of problems the problem of finding the optimal measurement
can be simplified considerably and is easily evaluated. 

It is a very useful consequence of our analysis that for any cryptographic application that makes use of such encodings, we can always perform a
relabeling of states $\rho_{xb}$ such that post-measurement information becomes useless. More precisely, we will associate $\vec{x}$ with the new
all $(0,\ldots,0)$ vector and $\vec{\underline{x}}$ with the new $(1,\ldots,1)$ vector. That is, for the optimal vector $\vec{x}$ we let
\begin{align}
\rho^{\rm new}_{0b} &:= \rho_{x^{(b)}b}\ ,\\
\rho^{\rm new}_{1b} &:= \rho_{(1 - x^{(b)})b}\ .
\end{align}
Clearly, by Lemma~\ref{lem:tightBound} we then have for $\ens^{\rm new} = \{\rho_{xb}^{\rm new}\}_{xb}$ 
that 
\begin{align}
\Delta(\ens^{\rm new},P) = 0\ ,
\end{align}
as desired.

\subsection{Example}

We now consider a small example that illustrates how our statement applies to the case where we have only two possible encodings $\cB = \{0,1\}$ into two orthogonal pure states in 
dimension $d=2$, and we choose the encoding uniformly at random ($p_b = 1/2$). 
A simple example is encoding into the BB84 bases~\cite{bb84}, where we pick the computational basis for $b=0$ and the Hadamard basis for $b=1$.
We now show that in two dimensions, post-measurement information is useless if and only if the angle between the Bloch vectors for the states 
$\rho_{00}$ and $\rho_{01}$ 
obeys $\theta \leq \frac{\pi}{2}$ as illustrated in Figures~\ref{fig:piUseless} and~\ref{fig:piUseful}.

\begin{figure}
\includegraphics{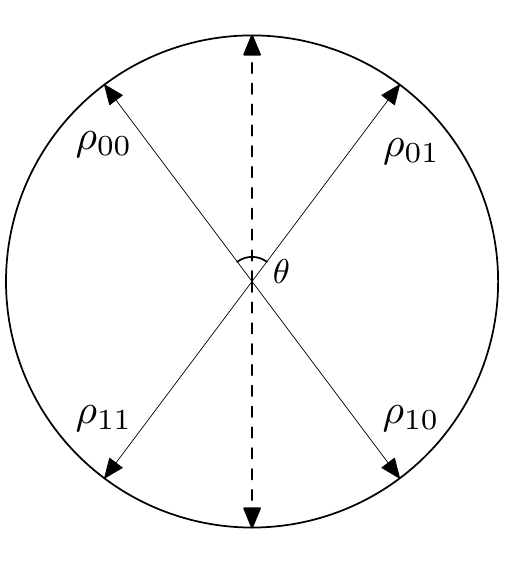}
\caption{Post-measurement information is useless iff $\theta \leq \frac{\pi}{2}$. The dashed line corresponds to the Bloch
vector of the optimal measurement using post-measurement information consisting of two rank one projectors $M_{00}$ and $M_{11}$, which is the same measurement one would make
for standard state discrimination. We output the same bit, no matter what encoding information $b$ we receive.}
\label{fig:piUseless}
\end{figure}

\begin{figure}
\includegraphics{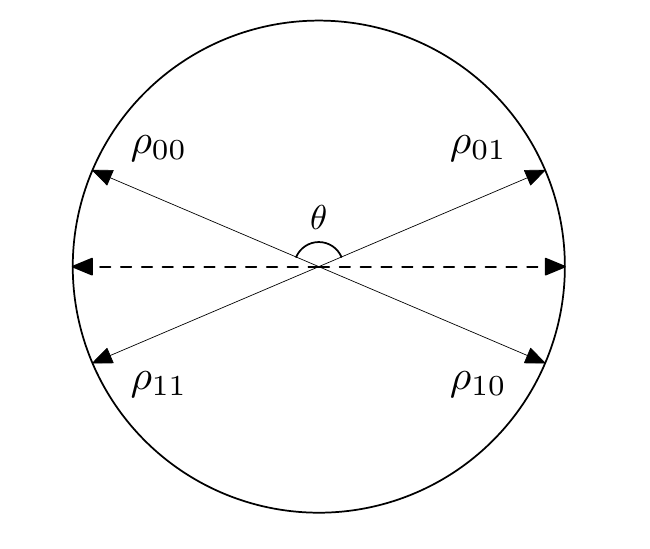}
\caption{Post-measurement information is useful for $\theta > \frac{\pi}{2}$. The dashed line corresponds to the Bloch vector of the optimal measurement using post-measurement information
consisting of two rank one projectors $M_{01}$ and $M_{10}$, which is the measurement one would make in standard state discrimination, if we were to distinguish $(\rho_{00} + \rho_{11})/2$ from
$(\rho_{01} + \rho_{10})/2$. Which bit we output depends on the post-measurement information we receive.}

\label{fig:piUseful}
\end{figure}
Note that in this example the average states are given by
\begin{align}
\rho_{(0,0)} &= \frac{1}{2}\left(\rho_{00} + \rho_{01}\right)\ ,\label{eq:rho00}\\
\rho_{(1,1)} &= \frac{1}{2}\left(\rho_{10} + \rho_{11}\right)\ ,\label{eq:rho11}\\
\rho_{(0,1)} &= \frac{1}{2}\left(\rho_{00} + \rho_{11}\right)\ ,\label{eq:rho01}\\
\rho_{(1,0)} &= \frac{1}{2}\left(\rho_{10} + \rho_{01}\right)\ .\label{eq:rho10}
\end{align}
The two partitions we are considering are $\tilde{T}_{(0,0)} = \{(0,0), (1,1)\}$ and $\tilde{T}_{(0,1)} = \{(0,1),(1,0)\}$.
Let $\vec{v}_0$ and $\vec{v}_1$ be the Bloch vectors corresponding to the states $\rho_{00}$ and $\rho_{01}$ respectively.
We have from
Lemma~\ref{lem:eigenvalues} that 
\begin{align}
\lambda_{\rm max}(\rho_{\vec{x}}) &= 
\lambda_{\rm max}(\rho_{\vec{\underline{x}}})  \\
&=\left\{\begin{array}{cc}
\frac{1}{2}\left(\id + \|v_0 + v_1\|_2\right) & \mbox{ for } \vec{x} = (0,0)\ ,\\[2mm]
\frac{1}{2}\left(\id + \|v_0 - v_1\|_2\right) & \mbox{ for } \vec{x} = (0,1)\ .
\end{array}
\right.
\end{align}
Hence, by Lemma~\ref{lem:tightBound} post-measurement information is useless if and only if
\begin{align}\label{eq:blochCondition}
\|v_0 + v_1\|_2 \geq \|v_0 - v_1\|_2\ .
\end{align}
Since $\|v_0\|_2 = \|v_1\|_2 = 1$ for pure states, we have $\|v_0 + v_1\|_2 = 2 + 2 \cos\theta$ and $\|v_0 - v_1\|_2 = 2 - 2 \cos \theta$
and thus~\eqref{eq:blochCondition} holds if and only if $\theta \leq \frac{\pi}{2}$. 
The optimal measurement is again given by~\eqref{eq:meas}.
Note that this is rather intuitive, since for partition $\tilde{T}_{(0,0)}$ we always give the same answer, no matter what post-measurement information we receive.

\section{Classical ensembles}\label{sec:classicalEnsembles}

We saw above that for the case of Clifford encodings even if post-measurement information was useful for the original problem, that is, $\psuc(\ens,P) < \psucPI(\ens,P)$, 
we could always perform a relabeling to obtain a new problem for which post-measurement
information is useless. We now show that this is a unique quantum feature, and is not present in analogous classical problems as
long as we are able to gain some information even without post-measurement information, i.e., $\psuc(\ens,P) > 1/|\cX|$. 
We thereby call a problem \emph{classical} if and only if all states $\rho_{xb}$ commute. 

We again focus on the case where we wish to encode 
a single bit $x \in \01$.
Let $\Pi_{xb}$ be a projector onto the support of $\rho_{xb}$. For simplicity, we will assume in the following
that $\Pi_{0b} + \Pi_{1b} = \id$ for all encodings $b$, and that the projectors are of equal rank $r = \rank(\Pi_{0b}) = \rank(\Pi_{1b})$.
We also assume that $\rho_{xb} = \Pi_{xb}/r$. It is straightforward to extend our argument to a more general case, but makes it more difficult
to follow our idea.

In~\cite[Lemma 5.1]{ww:pistar} it was shown that if $[P_{xb},P_{x'b'}] = 0$ for all bits $x,x'$ and encodings $b,b'$ of this form 
\begin{align}
\psucPI(\ens,P) = 1\ .
\end{align}
Recall that we are interested in the case where $\psuc(\ens,P) < \psucPI(\ens,P)$. Hence, our goal will be to show that
there exists no relabelling as in the previous section that allows us to create a new problem $\ens^{\rm new}$ for which
$\psuc(\ens^{\rm new},P) = \psucPI(\ens,P) = 1$.

\subsection{Non-local games}

To show our result, we will need the notion of non-local games which are a different way of looking at Bell inequalities~\cite{bell}.
For example, the well-known CHSH inequality~\cite{chsh} takes the following form when converted to a game. 
Imagine two space-like separated parties, Alice and Bob.
We choose two questions $s, t \in \01$ uniformly at random and send them to Alice and Bob respectively. The rules are that they win the game if and only
if they manage to return answers $a,b \in \01$ such that $s \cdot t = a + b \mod 2$. Without loss of generality, we may thereby assume that
Alice and Bob perform a measurement depending on the question they receive, and simply return the outcome of that measurement.
To help them win the game, Alice and Bob may thereby agree on any shared state and measurements ahead of time, but are no longer able to communicate
once the game starts. The average probability that they win the game is thus
\begin{align}
p_{\rm win} = \max \frac{1}{4}\sum_{s,t} \sum_{\substack{a,b\\a + b \mod 2 = s \cdot t}} \Pr[a,b|s,t]\ , 
\end{align}
where $\Pr[a,b|s,t]$ is the probability that they return answers $a$ and $b$ given questions $s$ and $t$, and the maximization
is over all states and measurements allowed in a particular theory. Classically, we have
\begin{align}\label{eq:CHSHclassical}
p_{\rm win}^{\rm classical} = \frac{3}{4}\ .
\end{align}
In a quantum world, however, Alice and Bob can achieve
\begin{align}
p_{\rm win}^{\rm quantum} = \frac{1}{2} + \frac{1}{2 \sqrt{2}} \approx 0.853\ .
\end{align}
More general non-local games are of course possible, where we may have a larger number of questions and answers, and the rules
of the game may be more complicated. 

Of central importance to us will be the fact that if Alice's (or Bob's) measurements commute, then there exists a classical strategy that
achieves the same winning probability (see e.g.~\cite{steph:diss}). We now use this fact to prove our result.

\subsection{A classical-quantum gap}

To explain the main idea behind our construction, we focus on the case where we only have two possible encoding $L=2$.
That is, $\cB = \01$ and $\cX = \01$. We also assume that the bit $x$, as well as the encoding $b$ is chosen uniformly and independently 
at random. The states defining our problem are thus $\rho_{00}$, $\rho_{01}$, $\rho_{10}$ and $\rho_{11}$.
We again consider the two partitions labeled by $\vec{x} \in \01^2$ given by
\begin{align}
\tilde{T}_{(0,0)} &= \{(0,0),(1,1)\}.\\
\tilde{T}_{(0,1)} &= \{(0,1),(1,0)\}.
\end{align}
As before, we can associate a standard state discrimination problem with each of these partitions. For the first partition $\tilde{T}_{(0,0)}$
as wish to discriminate between the states $\rho_{(0,0)}$ and $\rho_{(1,1)}$ specified by~\eqref{eq:rho00} and~\eqref{eq:rho11} where we are given one of the
two states with equal probability. Let $p_1$ denote the success probability
of solving this problem, maximized over all possible measurements. 
Note that our condition of being able to gain some information in the state discrimination problem corresponds to having
\begin{align}\label{eq:lowerBoundP1}
\frac{1}{2}< p_1\ .
\end{align}
For the second partition $\tilde{T}_{(0,1)}$, we wish to discriminate between $\rho_{(0,1)}$ and $\rho_{(1,0)}$ from~\eqref{eq:rho01} and~\eqref{eq:rho10}, again
given with equal probability. Let $p_2$ denote the corresponding success probability for the second partition.
Note that since we have only two possible partitions here constructed in the way outlined in Section~\ref{sec:tightBounds}, our goal of showing
that there exists no relabeling that makes post-measurement information useless can be rephrased as showing that $p_2 < 1$.

We now show that these two state discrimination problems arise naturally in the CHSH game.
In particular, we show in the appendix that
\begin{lemma}
There exists
a strategy for Alice and Bob to succeed at the CHSH game with probability $(p_1 + p_2)/2$, where Alice's
measurements are given by the projectors $\{P_{00},P_{10}\}$ and $\{P_{01},P_{11}\}$.
\end{lemma}

However, recall that if the ensemble of states is classical the projectors $P_{xb}$ all commute, and hence there exists a classical
strategy for Alice and Bob that also achieves a winning probability of $(p_1 + p_2)/2$. Hence, by~\eqref{eq:CHSHclassical} we must
have
\begin{align}
\frac{p_1 + p_2}{2} &\leq \frac{3}{4}\ .
\end{align}
Using~\eqref{eq:lowerBoundP1} this implies $p_2 \leq 3/2 - p_1 < 1 = \psucPI$,
and hence the relabelling corresponding to the second partition
cannot make post-measurement information useless. To summarize we obtain that~\footnote{Any relabeling that relabels at least one $\rho_{xb}$ is called non-trivial.}

\begin{theorem}
For the case of two encodings of a single bit chosen uniformly at random (i.e., $p_{xb} = 1/4$), which do allow us to gain
some information even without post-measurement information ($\psuc > 1/2$), 
there exists no non-trivial relabeling that renders post-measurement information useless.
\end{theorem}

Note that if we are able to gain some information in both state discrimination problems, i.e., $p_1, p_2 > 1/2$ the preceding discussion also implies that
$p_1, p_2 < 1$, that is, post-measurement information is never useless. Bounds on Bell inequalities corresponding to bounds on the maximum winning probability that
can be achieved in a classical world can thus allow us to place bounds on how well we can solve state discrimination problems \emph{without} post-measurement information.

This is in stark contrast to the quantum setting. For example, for the BB84 encodings it is not hard to see that $p_1 = p_2 = \psucPI \approx 0.853$~\cite{ww:pistar}, 
and hence post-measurement information is always useless. Yet, there exist classical encodings~\cite{ww:pistar} for which $p_1 = p_2 = 3/4$ but $\psucPI  = 1$.

To analyze the case of multiple encodings, we have to consider more complicated games than the one obtained from the CHSH inequality.
A natural choice is to consider games in which Bob has to solve different state discrimination problems corresponding to different partitions
of the vectors $\vec{x}$ depending on his question $t$ in the game. To make a fully general statement we would like to include all possible partitions. Clearly,
however the above approach can also be used to place bounds on the average of success probabilities for a subset of partitions by defining a game with less questions, 
and evaluating it's maximum classical winning probability.

\section{Conclusions}

Our work raises several immediate open questions. First of all, can we obtain sharper bounds? Since solving an SDP numerically is still very expensive in higher dimensions, it would
also be interesting to prove bounds on how well generic measurements such as the square-root measurement (also known as the pretty good measurement~\cite{wootters:pgm}) perform.
The pretty good measurement is a special case of Belavkin's weighted measurements~\cite{belavkin:optimal,belavkin:radio,mochon:pgm}, which was already used in its cube weighted form in~\cite{ww:pistar} 
to provide bounds on the state discrimination with post-measurement information. Such bounds have most recently been shown by Tyson~\cite{tyson:pgm} for standard state discrimination.
Yet, no good bounds are known on how well such measurements perform for our task.
More generally, it would be very interesting to see whether one can adapt the iterative procedures investigated in~\cite{werner:iterate,jezek:iterate,jezek:iterate2,tyson:oldIterate} 
to find optimal measurements for the case of standard state discrimination without post-measurement information 
to this setting. Concerning such iterative procedures, we would like to draw special attention to the recent work by Tyson~\cite{tyson:newIterate} generalizing monotonicity 
results for such iterates~\cite{reimpell:diss}, which could be applied here.

Naturally, it would be very interesting to know if our results for Clifford encodings can be extended to a more general setting.
Our discussion of classical ensembles shows that there exist problems for which $\psuc < \psucPI$ no matter what relabeling we perform~\cite{ww:pistar},
and hence we cannot hope that a similar statement holds in general. Nevertheless, it would be interesting to obtain necessary and sufficient conditions
for when post-measurement is already useless, or otherwise can be made useless by performing a relabeling.

\acknowledgments

DG thanks John Preskill and Caltech for a Summer Undergraduate Research Fellowship.
SW thanks Robin Blume-Kohout and Sarah Croke for interesting discussions.
SW is supported by NSF grants PHY-04056720 and PHY-0803371.

\appendix

\bigskip
In this appendix, we provide the technical details of our claims. For ease of reading, we thereby provide the proofs together with the statement of the lemmas.

\section{Proofs of Section~\ref{sec:bounds}}
\subsection{Optimality conditions}
\begin{lemma}
A POVM with operators $\{M_{\vec{x}}\}_{\vec{x}}$ is optimal for state discrimination with post-measurement information for the ensemble
$\cE = \{p_{xb}, \rho_{xb}\}$ if and only if the following two conditions hold:
\begin{enumerate}
\item $Q := \sum_{\vec{x}} \tau_{\vec{x}} M_{\vec{x}}$ is Hermitian.
\item $Q \geq \tau_{\vec{x}}$ for all $\vec{x} \in \cX^{\times L}$.
\end{enumerate}
\end{lemma}
\begin{proof}
Suppose first that the two conditions hold. Note that condition (2) tells us that $Q$ is a feasible solution, that is, it satisfies all constraints
for the dual SDP. By weak duality of SDPs we thus have $v_{\rm primal} \leq v_{\rm dual} \leq \tr(Q)$, and from condition (1)
we also have that $\tr(Q) = \sum_{\vec{x}} \tr(M_{\vec{x}} \tau_{\vec{x}}) \leq v_{\rm primal}$. Hence the POVM forms an optimal solution for the SDP.

Conversely, suppose that $\{M_{\vec{x}}\}_{\vec{x}}$ is an optimal solution for the primal SDP.
Let $Q$ be the optimal solution for the dual SDP.
Note that this means that $Q$ already satisfies condition (2), and all that remains is to show that $Q$ has the desired form given by condition (1).
Since $M_{\vec{x}} = \id/|\cX^{\times L}|$ is a feasible solution for the primal SDP, we have by Slater's condition~\cite{boyd:book} that the optimal
values $v^*_{\rm primal}$ and $v^*_{\rm dual}$ are equal, i.e., $v^*_{\rm primal} = v^*_{\rm dual}$. 
Using the fact that $\sum_{\vec{x}} M_{\vec{x}} = \id$ and that the trace is cyclic we thus have
\begin{align}
\tr(Q) - \sum_{\vec{x}} \tr(M_{\vec{x}} \tau_{\vec{x}}) &= \sum_{\vec{x}} 
\tr((Q - \tau_{\vec{x}})M_{\vec{x}}) = 0\ ,
\end{align}
Since $Q \geq \tau_{\vec{x}}$ (equivalently $Q - \tau_{\vec{x}} \geq 0$.), 
and $M_{\vec{x}} \geq 0$ for all $\vec{x}$ we have that all the terms $\tr((Q-\tau_{\vec{x}})M_{\vec{x}})$ in the sum are positive and hence we must have
for all $\vec{x}$ that $\tr((Q - \tau_{\vec{x}})M_{\vec{x}}) = 0$. Again using the fact that the two operators are positive semidefinite, and the cyclicity of the trace
we thus have for the optimal solution that
\begin{align}
(Q - \tau_{\vec{x}})M_{\vec{x}} = 
M_{\vec{x}}(Q - \tau_{\vec{x}}) = 0
\end{align}
Summing the l.h.s. over all $\vec{x}$ and noting that $\sum_{\vec{x}} M_{\vec{x}} = \id$ then gives us condition (1).
\end{proof}

\subsection{Upper bound}
\begin{lemma}
Let $N = |\cX|$ be the number of possible strings, and suppose that the joint distribution over strings and encodings satisfies
$p_{xb} = p_b/N$, where the distribution $\{p_b\}_b$ is arbitrary. Then
\begin{align}
\psucPI(\cE,P) \leq \frac{1}{N} \tr\left[\left(\sum_{\vec{x}} {\rho}_{\vec{x}}^\alpha\right)^{1/\alpha}\right]\ ,
\end{align}
for all $\alpha > 1$, where $\cE = \{\rho_{xb}\}_{xb}$, $P = \{p_{xb}\}_{xb}$ and ${\rho}_{\vec{x}} = \sum_{b=1}^L p_b \rho_{x_b b}$.
\end{lemma}
\begin{proof}
Note that since $y^{1/\alpha}$ is operator monotone for $\alpha > 1$~\cite[Theorem V.1.9]{bathia:book} we have
\begin{align}
\rho_{\vec{x}} = \left(\rho_{\vec{x}}^\alpha\right)^{\frac{1}{\alpha}} \leq \left(\sum_{\vec{x}} \rho_{\vec{x}}^{\alpha}\right)^{\frac{1}{\alpha}}\ .
\end{align}
Using the fact that $\sum_{\vec{x}} M_{\vec{x}} = \id$ we hence obtain
\begin{align}
\psucPI(\cE,P) &= \frac{1}{N} \sum_{\vec{x}} \tr\left(M_{\vec{x}} \rho_{\vec{x}}\right)\\
&\leq \frac{1}{N} \sum_{\vec{x}} \tr\left[M_{\vec{x}} 
 \left(\sum_{\vec{x}} \rho_{\vec{x}}^{\alpha}\right)^{\frac{1}{\alpha}}\right]\\
& = \frac{1}{N}  \tr\left[\left(\sum_{\vec{x}} \rho_{\vec{x}}^{\alpha}\right)^{\frac{1}{\alpha}}\right]\ ,
\end{align}
as promised.
\end{proof}

\subsection{Lower bound}

\begin{lemma}
The success probability \emph{with} post-measurement information is at least as large as the success probability of 
a derived problem \emph{without} post-measurement information, i.e.,
$$
\psucPI(\ens,P) \geq \max_{\vec{y}} \psuc(\ens_{T_{\vec{y}}},P)\ .
$$
\end{lemma}
\begin{proof}
This follows immediately from the discussion by noting that
\begin{align}
\sum_{\vec{x}} 
\tr\left(M_{\vec{x}}\rho_{\vec{x}}\right)
= \sum_{\vec{y}} \sum_{\vec{x} \in P_{\vec{y}}}
\tr\left(M_{\vec{x}}\rho_{\vec{x}}\right)\ .
\end{align}
\end{proof}

\section{Proofs of Section~\ref{sec:tightBounds}}
\subsection{Without post-measurement information}

\begin{lemma}\label{lem:Qsum}
For the measurement defined by~\eqref{eq:meas} we have
\begin{align}
Q & = 
\frac{1}{2}\left(\rho_{\vec{x}} M_{\vec{x}} + \rho_{\vec{\underline{x}}} M_{\vec{\underline{x}}}\right)
 =  \frac{1}{2d}\left(1 + \|\vec{v}_\vx\|_2\right)\id\ ,
\end{align}
and hence $Q$ is Hermitian.
\end{lemma}
\begin{proof}
We use the shorthand $\vec{a}_\vx \cdot \vec{\Gamma} = \sum_j a^{(j)}_\vx \Gamma_j$. We have
\begin{align}
\rho_{\vec{x}} M_{\vec{x}} &= \frac{1}{2d} \left(\id + (\vec{v}_\vx + \vec{a}_\vx)\cdot \vec{\Gamma} + (\vec{v}_\vx \cdot \vec{a}_\vx)\id\right)\ ,\\
\rho_{\vec{\underline{x}}} M_{\vec{\underline{x}}} &= \frac{1}{2d} \left(\id - (\vec{v}_\vx + \vec{a}_\vx)\cdot \vec{\Gamma} + (\vec{v}_\vx \cdot \vec{a}_\vx)\id\right)\ ,
\end{align}
where the equality follows from the fact that 
\begin{align}
(\vec{v}_\vx \cdot \vec{\Gamma})(\vec{a}_\vx \cdot \vec{\Gamma}) &= \frac{1}{2}\sum_{jk} v^{(j)}_\vx a^{(k)}_\vx \{\Gamma_j,\Gamma_k\}\ ,\\
&= (\vec{v}_\vx \cdot \vec{a}_\vx)\id\ .
\end{align}
Using that $\vec{v}_\vx \cdot \vec{v}_\vx = \|\vec{v}_\vx\|_2^2$ gives our claim.
\end{proof}

\begin{lemma}\label{lem:eigenvalues}
The largest eigenvalue of $\rho_{\vec{x}}$ and $\rho_{\vec{\underline{x}}}$ is given by
\begin{align}
\lambda_{\rm max}(\rho_{\vec{x}}) = \lambda_{\rm max}(\rho_{\vec{\underline{x}}}) = \frac{1}{d}\left(1 + \|\vec{v}_\vx\|_2\right)\ .
\end{align}
\end{lemma}
\begin{proof}
We now show that our claim for $\rho_{\vec{x}}$.
Our goal is to evaluate 
\begin{align}
\lambda_{\rm max}(\rho_{\vec{x}}) = \max_{\sigma} \tr(\sigma \rho_{\vec{x}})\ ,
\end{align}
where the maximization is taken over all states $\sigma$. Using the fact that
the set of operators $\{\id,\Gamma_j, i \Gamma_j \Gamma_k,\ldots\}_{jk\ldots}$ forms
an orthonormal (with respect to the Hilbert-Schmidt inner product) basis
for the $d \times d$ Hermitian matrices we can write
\begin{align}
\sigma = \frac{1}{d}\left(\id + \sum_j s^{(j)} \Gamma_j + \ldots \right)\ .
\end{align}
Since $\tr(\Gamma_j \Gamma_k) = 0$ for $j \neq k$, and we can rewrite $\rho_{\vec{x}} = \frac{1}{d}\left(\id + \sum_j a^{(j)}_\vx \Gamma_j\right)$ this gives us 
\begin{align}\label{eq:maxProblem}
\tr(\sigma \rho_{\vec{x}}) = \frac{1}{d}\left(1 + \vec{s} \cdot \vec{v}_\vx\right)\ ,
\end{align}
where $\vec{s} = (s^{(1)},\ldots,s^{(2n+1)})$ and $\cdot$ denotes the Euclidean inner product.
Since $\sigma \geq 0$ if and only if $\|\vec{s}\|_2 \leq 1$~\cite{ww:pistar} we have that the maximum in~\eqref{eq:maxProblem}
is attained for $\sigma = (\id + \sum_j s^{(j)} \Gamma_j)/d$ with
\begin{align}
\vec{s} = \frac{\vec{v}_\vx}{\|\vec{v}_\vx\|_2}\ , 
\end{align}
which gives our claim.
The argument for $\rho_{\vec{\underline{x}}} = \frac{1}{d}\left(\id - \sum_j a^{(j)}_\vx \Gamma_j\right)$ is analogous.
\end{proof}

\subsection{With post-measurement information}

\begin{lemma}
For our class of problems 
\begin{align}
\psucPI(\ens,P) = \max_{\vec{x}} \psuc(\ens_{\tilde{T}_{\vec{x}}})\ ,
\end{align}
and post-measurement information is useless if and only if the maximum on the r.h.s. is attained by $\vec{x} = (0,\ldots,0)$.
\end{lemma}
\begin{proof}
Let $\vec{x}$ be the string that achieves the optimum on the r.h.s of~\eqref{eq:rhs}.
We now claim that $Q = \frac{1}{2}\left(\rho_{\vec{x}} M_{\vec{x}} + \rho_{\vec{\underline{x}}} M_{\vec{\underline{x}}}\right)$
is an optimal solution to the SDP for the problem of state discrimination \emph{with} post-measurement information.
First of all, note that Lemma~\ref{lem:Qsum} gives us that $Q$ is Hermitian. We then have by Lemma~\ref{lem:eigenvalues}
that $Q \geq \frac{1}{2}\rho_{\vec{\tilde{x}}}$ \emph{for all} possible $\vec{\tilde{x}}$. Our claim now follows from Lemma~\ref{lem:conditions},
and by noting that for the partition $\vec{x} = (0,\ldots,0)$ we will always give the same answer, no matter what post-measurement information we receive later on.
\end{proof}
\end{document}